\documentclass[11pt]{article}
\usepackage{graphicx,amsthm,amsmath,fullpage}
\newtheorem{thm}{Theorem}
\newtheorem{lm}{Lemma}
\newtheorem*{rem}{Remarks}
\newtheorem{claim}{Claim}
\newtheorem{claim2}{Claim}

\begin{document}
  \title{The complexity of nonrepetitive edge coloring of graphs}
  \author{Fedor Manin
    \thanks{This research was conducted as part of the Summer Undergraduate
    Research Fellowship program at Caltech and partly supported by grant NSF
    CCF-0346991.}}
  \maketitle
  \begin{abstract}
    A squarefree word is a sequence $w$ of symbols such that there are no
    strings $x, y$, and $z$ for which $w=xyyz$.  A nonrepetitive coloring of a
    graph is an edge coloring in which the sequence of colors along any open
    path is squarefree.  The Thue number $\pi(G)$ of a graph $G$ is the least
    $n$ for which the graph can be nonrepetitively colored in $n$ colors.  A
    number of recent papers have shown both exact and approximation results for
    Thue numbers of various classes of graphs.  We show that determining
    whether a graph $G$ has $\phi(G) \leq k$ is $\mathbf{\Sigma}_2^p$-complete.

    When we restrict to paths of length at most $n$, the problem becomes
    \textbf{NP}-complete for fixed $n$.  For $n=2$, this is the edge coloring
    problem; thus the bounded-path version can be thought of as a
    generalization of edge coloring.
  \end{abstract}

  \section{Introduction}

  The study of avoiding repetition in combinatorial structures stems from a
  1906 paper by Axel Thue \cite{thue}, which showed that a there is an infinite
  word in $\{0,1,2\}$ that has no subword of the form $xx$ for some finite
  word $x$.  Such a word is called squarefree.

  A paper by Alon, Grytczuk, Haluszczak, and Riordan \cite{alon} introduced the
  generalization to edge colorings of graphs.  An edge coloring of a graph is
  \emph{nonrepetitive} if for any open path through the graph, the pattern of
  colors along the path is squarefree.  The Thue number $\pi(G)$ is the
  smallest number $k$ for which $G$ can be nonrepetitively colored with $k$
  colors.  Alon et al.~also mentioned the possibility of studying \emph{vertex
    nonrepetitive colorings}, i.e., vertex colorings with the same property.
  This notion has been further studied in papers by Czerwi\'nski and Grytczuk
  \cite{czer} and by Bar\'at and Varj\'u \cite{barat1}.  Other similar notions
  involving walks rather than paths have been developed by Bre\v sar and Klav\v
  zar \cite{bresar} and again by Bar\'at and Varj\'u \cite{barat2}.  For a more
  detailed survey of results regarding vertex nonrepetitive colorings, we refer
  the reader to Grytczuk \cite{gry}.

  The many papers in this area employ various notations and terminologies; we
  will stick to the original terminology of Alon et al.~whenever possible.
  Therefore, we refer to edge nonrepetitive colorings simply as nonrepetitive
  colorings and denote the edge coloring Thue number as $\pi(G)$.  We call a
  path \emph{square} if it is of the form $xx$ for some color-string $x$.

  Our goal is to study nonrepetitive colorings from the point of view of
  computational complexity.  Whatever version of $\pi(G)$ is used, deciding
  whether $\pi(G) \leq k$ is an $\exists\forall$ problem: does there exist a
  $k$-coloring of the edges or vertices of $G$ such that the pattern of colors
  along any open path, or one of the kinds of walks, is squarefree?  Such
  questions usually belong in the class $\mathbf{\Sigma_2^p}=\mathbf{NP^{NP}}$,
  a class in the second level of the polynomial hierarchy.  This class and
  problems known to be complete for it are described in a survey by Schaefer
  and Umans \cite{umans1}.

  Many of the results of papers on nonrepetitive colorings which bound the Thue
  number of various classes of graphs---for example, Alon et al.\cite{alon}
  show that $\pi(G) \leq c\Delta^2$ for constant $c$ for any graph of degree
  $\Delta$.  On the other hand, Alon et al.~also show that for any $\Delta$
  there is a graph $G$ of maximum degree $\Delta$ such that any nonrepetitive
  vertex coloring of $G$ requires at least $c\Delta^2/\log\Delta$ colors.
  Both these results, and many others like them, are probabilistic and thus
  give us little information about specific graphs.

  The complexity of vertex nonrepetitive colorings has recently been studied by
  Marx and Schaefer \cite{ms}, who showed that deciding whether such a coloring
  is nonrepetitive is \textbf{coNP}-complete.  In this paper we show that the
  corresponding problem for edge colorings is \textbf{coNP}-complete.  We then
  show that deciding whether $\pi(G) \leq k$ is $\mathbf{\Sigma_2^p}$-complete.

  Finally, we turn to colorings that are nonrepetitive for bounded-length
  paths.  Let $\chi_i^e(G)$ be the least number of colors required to color the
  edges of $G$ so that no open path contains the square of a pattern of colors
  of length $\leq i$, and let $\chi_i(G)$ be the least number of colors
  required to color the vertices of $G$ under analogous conditions.  Since the
  number of paths of length $2i$ is bounded by $\lvert V \rvert^{2i}$, for
  fixed $i$ each such problem is contained in \textbf{NP}.  Then $\chi_1(G)
  \leq k$ is simply \textsc{graph coloring} and $\chi_1^e(G) \leq k$ is
  \textsc{edge coloring}, shown to be \textbf{NP}-complete by Holyer in 1981
  \cite{holyer}.  $\chi_2(G)$ is known as the symmetric chromatic number; the
  corresponding decision problem is shown to be \textbf{NP}-complete by Coleman
  and Mor\'e \cite{cm}.  We show that the decision problems corresponding to
  $\chi_2^e(G)$ and $\chi_{2k}^e(G)$ for any $k \geq 4$ are
  \textbf{NP}-complete.

  \section{Methods}
  The nonrepetitive colorability of a graph with $k$ colors is a rather
  slippery property.  In the construction of a nonrepetitive coloring, changing
  the color of an edge may produce a square path of arbitrary length and thus
  affected by colors of edges arbitrarily far away.  Thus the local properties
  of a graph cannot guarantee that it is nonrepetitively colorable.  On the
  other hand, if a graph does not have a nonrepetitive coloring, then we must
  examine all of the colorings and find a square path in each, which also feels
  like a hard problem.  Thus neither direction of the reduction is easy.

  There are two immediately visible approaches to resolving this dilemma.
  Perhaps we could focus on fairly sparse graphs, thus minimizing the number of
  paths we must make sure are nonrepetitive.  But this would greatly increase
  the number of colorings we must consider.  We thus use this approach only in
  determining whether a preassigned coloring is nonrepetitive.  On the other
  hand, we could focus on graphs in which the degree of every vertex is close
  to the maximal degree of the graph.  Although this allows us to easily
  discard most possible colorings, the number of distinct paths grows
  exponentially.  Additionally, it becomes harder to create graphs which differ
  subtly in global structure in such a way that some of them are
  nonrepetitively colorable and others are not.

  In our proof, we combine these two approaches, producing graphs that are
  locally dense, in the sense that the degree of most vertices is large and
  close to the maximal degree, but globally sparse, in the sense that long
  paths must traverse a number of bottlenecks with only a few connections.
  While a pattern of local structures, or gadgets, is common to many if not all
  graph-theoretic complexity proofs, in this case in particular it allows us to
  partially isolate each portion of the problem: once we have established that
  the local structure is nonrepetitively colorable in a very constrained set of
  ways, we can us this to determine whether the global structure is
  nonrepetitively colorable.

  We will first construct the global structure, which is similar to the
  well-known reduction which shows that \textsc{hampath} is
  \textbf{NP}-complete.  The local structure, added later, will be composed of
  several types of dense subgraphs.  The simplest of these are the cliques of
  size $2^n$, which Alon et al.~show to be nonrepetitively colorable with $2^n$
  colors.  We also use the $n$-dimensional hypercube graphs, which have a
  number of useful properties which are discussed in the next section.
  Finally, in a number of cases we will \emph{saturate} a vertex $v$ by adding
  extra vertices only joined to $v$, called the \emph{plume} of $v$, to
  artificially inflate the degree.  This will allow us to more effectively use
  the saturation lemma below.

  \section{Preliminaries}
  We will now prove some results about nonrepetitive colorings that will aid in
  the complexity proofs.

  \begin{lm}[Saturation lemma]
    Let a vertex of a graph be \emph{saturated} if it is of maximal degree, and
    let a \emph{diamond} be a cycle of four vertices.  Suppose two opposing
    vertices of a diamond, $A$ and $B$, in a graph of degree $k$, are saturated
    and the saturating edges are not directly connected.  Then the graph may
    only be $k$-nonrepetitively colored if the diamond is colored with exactly
    two colors.
  \end{lm}
  \begin{proof}
    Let $a$ and $b$ be the colors of one path from $A$ to $B$ through the
    diamond.  Then, since $A$ and $B$ are saturated, there must be an edge of
    color $b$ at $A$ and an edge of color $a$ at $B$, forming a path of colors
    $abab$.  If this path is not a cycle, then the coloring is not
    nonrepetitive.  Thus, these colors must form the other half of the
    diamond.
  \end{proof}
  
  \begin{lm}[Properties of hypercubes]
    The nonrepetitive $k$-edge-coloring of a $k$-hypercube exists, is unique
    up to permutation, and has the following properties:
    \begin{enumerate}
    \item[(1)] The shortest path between any two points consists of distinct
      colors.
    \item[(2)] Any permutation of the colors of a path between two points
      consisting of distinct colors also colors a path between those two
      points, and no other sequence of distinct colors does.
    \item[(3)] Any path consisting of distinct colors is a shortest path
      between its endpoints.
    \item[(4)] There are ${k \choose i}$ vertices at distance $i$ from a given
      vertex $V$.
    \end{enumerate}
    The uniqueness of coloring also holds for portions of a $k$-hypercube whose
    distance from a chosen base vertex $V$ is at most $m$, which we call the
    \emph{first $m$ layers} of the hypercube.
  \end{lm}
  \begin{proof}
    We will identify the vertices of the hypercube with $\{0,1\}^k$, where
    there is an edge between $(\chi_1,\ldots,\chi_k)$ and $(\chi_1,\ldots,
    \chi_{i-1},1-\chi_i,\chi_{i+1},\ldots,\chi_k)$ for every $i$.

    Without loss of generality, let us choose the vertex $V$ identified with
    $\bar{0}$ to be a base vertex.  $V$ is saturated, so the coloring of the
    edges adjacent to it is unique up to permutation.  Thus the coloring of the
    first layer is unique.  Now suppose that we have colored the first $m$
    layers.  Every vertex is saturated and any two adjacent edges form a
    diamond (a face.)  Hence by the previous lemma the coloring of the $m$th
    layer determines the coloring of the $m+1$st layer.  By induction, the
    coloring of the first $m$ layers is unique, and so is the coloring of the
    entire hypercube.

    To show that a coloring exists, let us use the color $i$ for all edges
    for which $\chi_i$ changes between the two vertices.  For every $i$, the
    hypercube consists of two $k-1$-hypercubes, one in which the $i$th
    coordinate is 0 for all vertices and one in which it is 1.  Therefore, if
    we have a path containing two edges of color $i$, then we can create an
    alternate path with the same start and end vertices by removing these
    two edges and reversing $\chi_i$ for the vertices in between them.

    Now suppose we have a square path.  Then for every dimension, there is an
    even number of edges in that dimension on the path, and we can remove all
    these edges by pairs from the path.  Thus the path is equivalent to the
    empty path, is therefore a cycle, and all open paths are squarefree.

    We will now use this construction to prove the rest of the properties.  For
    (1), suppose that we have a shortest path between two points that contains
    two edges of the same color.  Then we can remove these two edges and get a
    shorter path between the same two vertices.  Therefore the shortest path
    must consist of distinct colors.

    For (2), clearly a permutation of the changes to variables leads to the
    same result.  On the other hand, for any two shortest paths whose sets of
    colors differ, there is an $i$ such that $\chi_i$ is different between the
    two end vertices for one of the paths, but not for the other.  Thus the end
    vertices are different for the two paths.

    (3) follows directly from (1) and (2).

    A path of length $i$ represents a change in $i$ members of the $k$-tuple.
    There are ${k \choose i}$ possible such changes, proving (4).
  \end{proof}

  \section{Nonrepetitive coloring problems}

  In this section, we consider several related problems in which the coloring
  is entirely or almost entirely predetermined, and the main difficulty of the
  problem is in finding whether the coloring is nonrepetitive.  The proofs in
  this section are all related and provide a framework for the proof of the
  $\mathbf{\Sigma_2^p}$-completeness of \textsc{Thue number}.  We first
  consider colorings of directed graphs.

  \begin{thm}
    Given a directed graph $G$ and a coloring thereof, it is
    \textbf{coNP}-complete to determine whether the coloring is nonrepetitive.
    We call this problem \textsc{directed nonrepetitive coloring}.
  \end{thm}
  \begin{proof}
    We reduce from co-3SAT.  The reduction is similar to the standard reduction
    from 3SAT to \textsc{hampath}.  The edges are given a large number of
    colors, with 4 or fewer edges per color.  In effect, we restrict possible
    square paths to those that go through two edges of every color, mimicking
    the requirement that a Hamilton path must pass through every vertex.

    Let $f(x_1,\ldots,x_n)$ be an instance of 3SAT.  The reduction will use two
    types of gadgets: a variable gadget and a clause gadget.

    The variable gadget consists of two sets of $M$ vertices, where $M$ is the
    largest number of number of instances of a literal, together with a
    beginning vertex $b_i$.  A path of edges goes from $b_i$ through each of
    the vertices in a set to $b_{i+1}$, or a vertex $c$ if $i=n$.  The sets of
    vertices and paths going through them correspond to a true assignment and a
    false assignment; thus each set signifies a literal.

    Suppose $f$ has clauses numbered $C_1,\ldots,C_m$.  A clause gadget for
    clause $C_j$ has one vertex $e_j$, and we add an extra vertex $d$ and an
    edge from $c$ to $e_1$.  For each literal in $C_j$, we add an edge from
    $e_j$ to a vertex on the path corresponding to this literal and an edge
    from this vertex to $e_{j+1}$ or $d$ if $j=m$, making sure that this is the
    only clause gadget that uses this vertex.

    To complete the construction of the graph, we add a series of vertices
    $a_0,\ldots,a_{Mn+2m}$ and edges connecting $a_i$ to $a_{i+1}$ and
    $a_{Mn+2m}$ to $b_1$.  We call this structure the \emph{snout}.  We must
    now assign a coloring to this graph.  We will use a set $w_1,\ldots,
    w_{Mn+2m+1}$ of colors.  Starting from the tip, we color the $i$th edge of
    the snout with $w_i$.  We color the $k$th edge of a path through the
    $x_i$-gadget by $w_{M(i-1)+k}$; an edge going into $e_j$ by $w_{Mn+2j-1}$;
    an edge leaving $e_j$ by $w_{Mn+2j}$; and an edge going into $d$ by
    $w_{Mn+2m+1}$.

    Given a satisfying assignment to $f$, we can construct a square path by
    first traversing the snout, then taking the path through each variable
    gadget corresponding to the opposite of its assignment and the path through
    each clause gadget corresponding to a true literal, as illustrated for an
    example.  Conversely, if we have such a path, then taking the opposite of a
    traversed literal-path to be true gives us a satisfying assignment.
    \begin{figure}[t]
      \begin{center}
	\includegraphics{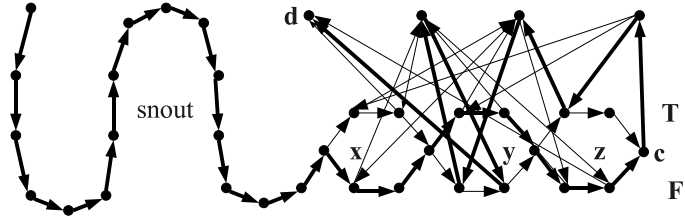}
      \end{center}
      \caption{The graph generated by $(x \vee y \vee z) \wedge (\neg x \vee
	\neg y \vee \neg z) \wedge (x \vee \neg y \vee \neg z)$, with a
	square path highlighted.}
    \end{figure}

    Now suppose that $f$ is unsatisfiable, and that there is a square path
    $p$, which then cannot have the form above.  Suppose first that the edges
    of $p$ are colored only by the colors of the variable gadgets.  Since the
    colors of the variable gadgets form the tip, but not the base of the snout,
    $p$ must consist only of the variable gadgets themselves; but this is
    impossible since the variable gadgets can only be traversed in the forward
    direction and thus a path through them has no two edges of the same color.

    (*) Thus $p$ must contain edges from clause gadgets.  Since any two edges
    in a clause gadget that share a color also share a vertex, an open path can
    contain only one of such a set.  Thus, for every color $a$ that colors an
    edge of $p$ contained in a clause gadget, $p$ contains an edge in the snout
    of color $a$.  Since $b_1$ is only the source of the two branches of the
    first variable gadget, one of these two branches must also be contained in
    the path.  Furthermore, since these branches share a vertex, only one of
    them may be contained in $p$.  The only other edge colored by $w_1$ is the
    edge at the very tip of the snout; thus the entire snout must be contained
    in $p$, so $p$ must contain a path belonging to each variable and clause
    gadget.  As above, this defines a satisfying assignment to $f$, a
    contradiction.

    Clearly, this reduction runs in polynomial, in fact quadratic, time.
  \end{proof}
  \begin{rem}
    The graph resulting from this reduction has maximum in- and out-degree 3.
    Each color is repeated at most four times.
  \end{rem}

  From this we can fairly easily show the same result for undirected graphs.

  \begin{thm}
    Given an undirected graph $G$ and an edge-coloring thereof, it is
    \textbf{coNP}-complete to determine whether the coloring is nonrepetitive.
    We call this problem \textsc{nonrepetitive coloring}.
  \end{thm}
  \begin{proof}
    We will now modify the previous reduction to work with undirected graphs.
    Let us add the colors $a_i^j,b_i^j,c_i^j,d_i^j$, for $i \in \{1,2,3\}$ and
    $j \leq Mn+2m+1$.  We call these the \emph{direction-determining colors}
    and the rest the \emph{original colors}.  We replace every edge inside the
    snout with the graph in figure 2 and every edge outside the snout with the
    path $a_i^jb_i^jw_jc_k^jd_k^j$ for some $i,k \in \{1,2,3\}$.  This can be
    done in such a way that there are no trivial repetitions, that is, no two
    edges of the same color are incident to the same vertex, since the maximum
    in- and out-degree of the directed graph is 3.  Clearly, this construction
    retains the property that each color colors no more than four edges, since
    none of the direction-determining colors need color more than two edges.
    We will show that this modification leaves the reduction valid.
    \begin{figure}[t]
      \begin{center}
	\includegraphics{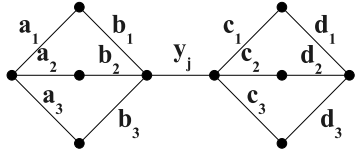}
      \end{center}
      \caption{The graph that replaces a snout edge of color $y_j$.}
    \end{figure}

    If $f$ has a satisfying assignment, we can traverse the graph as before,
    choosing the right path through the snout to produce the same $a_i$s,
    $b_i$s, $c_k$s, and $b_k$s as in the other half of the path.

    Now suppose that $f$ is unsatisfiable.  Suppose that the graph has a
    square path.  Clearly, this path must contain direction-determining colors.
    Each such color is represented only once in the non-snout edges, so we
    cannot make a square path without any snout edges.  This means that the
    paragraph marked (*) in the previous proof still applies, and we are done.
  \end{proof}

  We now modify the statement of the problem slightly in order to get a
  $\mathbf{\Sigma_2^p}$-complete problem.

  \begin{thm}
    Given an undirected graph $G$ and a set $\mathcal{S}$ of pairs $(e,S_e)$
    for each edge $e$, with $S_e$ a set of colors, it is $\Sigma_2^p$-complete
    to determine whether there exists a nonrepetitive coloring of $G$ that
    uses a color $s_e \in S_e$ for each $e$.  We call this problem
    \textsc{restricted Thue number}.
  \end{thm}
  In our case, $K_e$ will contain at most two colors for any edge, and only one
  color for most edges.
  \begin{proof}
    Clearly, the problem is in $\Sigma_2^p$.  To prove that it is
    $\Sigma_2^p$-hard, we shall slightly modify the previous reduction,
    reducing from co-$\forall\exists$3SAT.  Let $\forall x \exists y f(x,y)$ be
    an instance of $\forall\exists$3SAT.  Our graph $G$ will be the graph
    constructed for $\exists x,y\,f(x,y)$ in Theorem 2.  For each universally
    quantified variable, the two original edges that initially separate the
    variable gadget into a ``true'' and a ``false'' branch are given the sets
    $\{w_i^1\}$ and $\{w_i^0\}$ respectively, where $w_i$ is the color given
    these edges in the previous reduction.  The corresponding edge in the snout
    is given the set $\{w_i^t,w_i^f\}$.  The rest of the edges are given
    singletons containing the color they were colored in the previous
    reduction.

    Thus a coloring of the graph satisfying the restriction corresponds to an
    assignment of $x$, since it forces any potentially square path to go
    through the specified branch of the corresponding variable gadget.  This
    means, by the previous construction, that there is such a nonrepetitive
    coloring of $G$ iff there is an assignment to $x$ for which $\exists y f(x,
    y)$ is unsatisfiable.
  \end{proof}
  This proof provides the final incarnation of the global structure which we
  will later harness to show that \textsc{Thue number} is $\mathbf{\Sigma_2^p}
  $-complete.

  \section{Problems concerning bounded-length paths}

  In this section, we explore two problems in which we consider a polynomial
  number of possibly square paths, and therefore only the existential
  quantifier in the statement of the problems is over an exponentially large
  domain.  Together, the problems form a generalization of \textsc{edge
    coloring}, a problem which was shown to be \textbf{NP}-complete by Holyer
  in 1980 \cite{holyer}.

  In these proofs, we employ several tricks to find graphs whose small-diameter
  subgraphs are restricted in the ways they can be colored nonrepetitively.
  These tricks will later find applications in the proof of the
  $\mathbf{\Sigma_2^p}$-completeness of \textsc{Thue number}.  Thus although
  this material is not directly related to the main result, it builds some
  machinery for its proof.  Still, a reader interested only in the main proof
  may go on to the next section.

  \begin{thm}
    It is \textbf{NP}-complete to determine whether $\chi_2^e(G) \leq 6$.
  \end{thm}
  \begin{proof}
    We reduce from \textsc{edge coloring} for cubic graphs (graphs with 3 edges
    incident to every vertex) and 3 colors, which is shown to be
    \textbf{NP}-complete in \cite{holyer}.  The reduction consists of replacing
    each edge of graph $G$ with the graph depicted in Fig.~3, which we shall
    call a \emph{clam}.  The vertices incident to the edge are mapped to $A$
    and $B$.  We shall argue that the resulting graph $H$ can be 6-colored with
    nonrepetitive 4-paths iff $G$ can be 3-colored.  Suppose that we have such
    a coloring, and call the six colors $a,b,c,d,e,f$.
    \begin{figure}
      \begin{center}
	\includegraphics{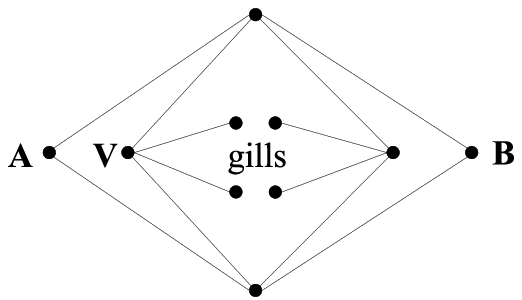}
      \end{center}
      \caption{A clam.}
    \end{figure}

    Note first that vertices $A$ and $B$ are both saturated, since they are
    images of vertices of $G$ of degree 3, so the outer diamond of the clam
    only has two colors---without loss of generality, $a$ and $b$.  Now suppose
    that one of the inner four edges (let us call them the \emph{gills}) is
    colored $a$.  Then if the color of an edge of the inner diamond connected
    to this edge is $c$, there is a path $acac$ through gill, inner diamond,
    outer diamond, and an edge of another clam.  Thus the gills must be colored
    differently from the outer diamond, as must the inner diamond since it is
    adjacent to the outer one.  This means that the vertex labeled $V$ and the
    one symmetric to it are saturated for the four colors other than $a$ and
    $b$, and thus the inner diamond must be colored with exactly two colors
    ($c$ and $d$.)

    Now let us add two more clams, incident to vertices $C$ and $D$
    respectively, at vertex $A$.  Suppose that the outer colors of the clams
    are $ce$ and $df$ respectively.  Then to prevent square paths of the type
    $acac$ from the inner diamond of clam $AC$ to the inner diamond of clam
    $AB$, the inner diamonds of clams $AC$ and $AD$ must consist of the colors
    $df$ and $ce$ respectively.  But then we get square paths of the form
    $dcdc$ from the inner diamond of clam $AC$ to the inner diamond of clam
    $AD$.  Thus the outer colors of $AC$ and $AD$ must be $cd$ and $ef$, in
    some order.  Furthermore, the inner colors of $AC$ must be $ef$ to prevent
    $acac$-type repetitions.  This forces the inner colors of $AD$ to be $ab$.
    This means that the coloring of one clam at a vertex of $G$ determines the
    sets of colors available to the other clams; by induction, each clam in the
    graph must be colored using one of these three patterns.

    This means that if we have a 3-coloring of $G$, we can map the three colors
    onto the three clam-coloring patterns to produce a 2-nonrepetitive
    6-coloring of $H$; conversely, with a nonrepetitive 6-coloring of $H$, we
    can map the three clam-coloring patterns onto three colors with which we
    can edge-color $G$.
  \end{proof}

  \begin{thm}
    For any natural number $k \geq 4$ and graph $G$, it is \textbf{NP}-complete
    to determine whether $\chi_{2k}^e(G) \leq 6k$.
  \end{thm}
  \begin{proof}
    This is a generalization of the previous reduction.  However, one must view
    the clam not as a unified gadget, but as a combination of elements of three
    gadgets.  The two edges incident to vertex $A$ belong to vertex gadget $A$,
    and likewise for vertex $B$.  The edge gadget proper, then, consists of the
    inner diamond and the gills.

    We again reduce from \textsc{edge coloring} for cubic graphs in 3 colors.
    Let $G$ be a cubic graph.  We associate to $G$ a 3-coloring $C(e)$ of the
    edges.  As we construct the gadgets in the reduction, we will also
    construct a $6k$-coloring of the edges based on $C(e)$.  Later, we will
    show that this coloring is $4k$-nonrepetitive iff $C(e)$ is a valid edge
    coloring, and that no other coloring of the edges is $4k$-nonrepetitive.
    We shall now enumerate the elements of the reduction, which reduces $G$ to
    a graph $H$.

    The vertex gadget consists of the first $k$ layers of a $6k$-hypercube with
    respect to a central vertex $K$.  By the hypercube lemma, this has a unique
    nonrepetitive coloring.

    Now let us divide the $6k$ colors into three groups of $2k$, called
    \emph{red}, \emph{green}, and \emph{blue}.  We shall use the term
    \emph{grue} to refer to the union of the sets blue and green.  The elements
    of the sets will be referred to as \emph{shades}.  For each of the groups,
    there are ${2k \choose k}$ vertices whose distance from $K$ is $k$ such
    that the path to them from $K$ consists of edges of colors from that group.
    We identify these vertices with an analogous group from another vertex
    gadget (with central vertex $L$) in order to form an edge.

    \begin{claim}
      In a $4k$-nonrepetitive coloring, the colors of the edges on either side
      of this connection are in the same set of $2k$.
    \end{claim}
    \begin{proof}
      Suppose we have a $4k$-nonrepetitive coloring that does not have this
      property.  Then take an arbitrary path from $K$ to $L$ of length $2k$,
      say with color pattern $r_1, \ldots, r_k, b_1, \ldots, b_k$.  By the
      properties of the $6k$-hypercube, there must be a path starting with $K$
      of pattern $b_k, \ldots, b_1$ and a path starting with $L$ of pattern
      $r_1, \ldots, r_k$.  This forms a square path of length $4k$, which must
      then be a cycle.  This means that the sets of colors on either side of
      the inter-vertex interface are the same.  Assume they are both red.
      Then, since the fully-red paths from $K$ to $L$ form a $2k$-hypercube,
      they can be colored nonrepetitively.
    \end{proof}

    We assume now that we are working with two vertices, $K$ and $L$, that are
    joined by a red edge.  The construction for the other two sets is
    analogous.

    It remains to add a gadget that will force the green and blue sets of one
    vertex to be the same as the green and blue sets of the other.  To this
    end, we take two vertices, $A$ and $B$, on the interface between two vertex
    gadgets, whose distance from each other is $2k$, to be two ends of a
    consistency gadget.

    The consistency gadget is a $4k$-hypercube modified in two ways.  First,
    every vertex whose distance from $A$ and $B$ is 2 or more has a
    \emph{plume} of $2k$ additional edges leading to disconnected vertices.
    A red edge separated by one edge from $A$ or $B$ generates a square path of
    length 4, so this means all the vertices of the hypercube are effectively
    saturated.  This implies that the hypercube is entirely grue and the plumes
    are red.  On the other hand, the farther red plume edges do not generate
    such a path, since it would have to repeat a red edge followed by two or
    more grue edges as a path from the red interface vertices, and thus
    required by hypercube property (3) to stray a distance more than $k$ from
    the central vertex of the vertex gadget.

    Next, for every path starting at $A$ or $B$ and consisting of $k$ edges of
    distinct shades of green, we remove the last edge.  We then do the same for
    blue.  We also remove all red plume edges incident to the vertices that are
    left unsaturated by this removal.  We call the $k$th layer of the hypercube
    from either $A$ or $B$ the \emph{gap layer}.

    \begin{claim}
      This consistency gadget is uniquely $4k$-nonrepetitively $6k$-colorable
      up to isomorphism.  This coloring is such that the hypercube edges are
      grue and colored as the restriction of a full hypercube, and the plumes
      are red.
    \end{claim}
    \begin{proof}
      Any diamond containing edges in the layer adjacent to the gap layer has
      two saturated vertices at equal distance from the vertex gadget.  Suppose
      there is an edge between two vertices that are both unsaturated due to
      the removal.  If those two vertices both became unsaturated from the
      removal of edges of the same color, then the edge would have to be of
      that color, and thus also removed.  On the other hand, an edge between
      two vertices that are unsaturated due to the removal of different colors
      would not exist by hypercube property (2).  Thus any diamond that has
      edges in the gap layer has saturated vertices either laterally or
      vertically, and so by the saturation lemma the coloring is a restriction
      of the full hypercube.
    \end{proof}

    We then add a second, identical consistency gadget at vertices $C$
    and $D$ whose distances from $A$ and $B$ are at least 4.   (This distance
    prevents a square path consisting of a red plume edge in the first
    consistency gadget, then some path through it, then two red edges, the same
    path through the second consistency gadget, and another red plume edge.)

    We refer to the set of colors that form the interface between two vertex
    gadgets as an \emph{edge set}.

    \begin{claim}
      If the original graph is not 3-edge-colorable, the resulting graph is not
      $4k$-nonrepetitively $6k$-edge-colorable.
    \end{claim}
    \begin{proof}
      Given a red edge set, suppose that another edge set of the same vertex
      gadget with central vertex $K$, belonging to edge $e$ of $G$, has both
      green and blue parts.  This means that that one of its consistency
      gadgets will have a path to it that contains both green and blue, since
      if one of the consistency gadgets has an all-blue side and an all-green
      side, then the other one does not.  Thus we can find a square path
      consisting of:
      \begin{enumerate}
	\item a red plume edge that duplicates the last edge of (3),
	\item a path of length $k$ in a consistency gadget of the red edge that
	  duplicates (4),
	\item a path from the consistency gadget of the red edge set to $K$,
	\item a mixed blue-green path from $K$ to one of the blue-green edge
	  set's consistency gadgets, and
	\item a path of length $k-1$ in this consistency gadget that duplicates
	  the first $k-1$ edges of (3).
      \end{enumerate}
      Therefore, the edge sets cannot be of mixed colors with respect to the
      red edge set's consistency gadget.
    \end{proof}

    \begin{claim}
      If $G$ is 3-edge-colorable, then the coloring of $H$ which we have
      described is $4k$-nonrepetitive.
    \end{claim}
    \begin{proof}
      Suppose that $G$ is 3-edge-colorable, and let $p=qr$ be a square path in
      $H$ of length at most $4k$, where $q$ and $r$ are the repetitions of the
      pattern.  The portion of $p$ inside a gadget $\mathcal{G}$ will be called
      $p_{\mathcal{G}}$.  Clearly, $p$ cannot be wholly inside a consistency
      gadget.  We have also already established that $p$ cannot be wholly
      outside any consistency gadgets.  Thus it include edges from both
      consistency and vertex gadgets.  Since a consistency gadget has distance
      $4k$ between its two attachment points, $p$ cannot have a consistency
      gadget in the middle.  In introducing the red plumes in the consistency
      gadget, we also showed that the whose middle vertex of $p$ cannot be a
      junction between a consistency gadget and a vertex.

      Thus at least one half-path of $p$ (without loss of generality, $q$)
      contains parts of both a consistency gadget $T$ and a vertex gadget $J$
      with center $K$.  $q_T$ contains at most one of the colors that are found
      on a shortest path to $K$ from the juncture with the consistency gadget.
      If such a color is in $q_T$, it must be a plume, which forces $q_T$ to be
      at least three edges long.

      Suppose now that the rest of $p$ is contained in $J$.  We will show that
      the distance from the last vertex $V$ of $p$ to $K$ is greater than $k$,
      a contradiction.  If a color has both of its repetitions inside $J$,
      these cancel each other out, so we are left with the edges that repeat
      $q_T$.  But since there is at least one more grue than red edge in $q_T$,
      this means that our initial distance $k$ from $K$ increases by at least
      one.

      This means that $p$ must access either a consistency gadget or another
      vertex gadget.  Without loss of generality, assume that $T$ belongs to a
      red edge set.

      \emph{Case 1}: $p$ accesses a second consistency gadget that belongs the
      same edge set.  Then there are at least four shades of red that are
      repeated an odd number of times in the vertex gadget portion of the path.
      Since there is at most one red edge in each consistency gadget part of
      the path, nothing can repeat these edges.

      \emph{Case 2}: $p$ accesses a different consistency gadget, without loss
      of generality from a blue edge set.  Then $p_J$ must contain at least
      $2k$ colors that color an odd number of edges, since the distance between
      the two consistency gadgets is $2k$.  This is the maximal number of
      distinct colors in the path we are considering, so each of the $2k$ must
      be repeated in one of the two consistency gadgets.  This means that one
      of the consistency gadgets contains a red-blue path of length at least
      $k$.  However, this must pass through edges that we have removed, so no
      such path exists.

      \emph{Case 3}: $p$ accesses a vertex gadget $J^\prime \neq J$ which
      interfaces $J$ at a non-red edge set.  This again means that there are
      $2k$ distinct colors in $p_J$, and the $2k$ edges in $p_T$ and
      $p_{J^\prime}$ must repeat these.  In particular this means that the edge
      that repeats the edge of $p_J$ adjacent to $T$ must be the edge of
      $p_{J^\prime}$ adjacent to $J$, and so this edge must be colored a shade
      of red, which is a contradiction.

      \emph{Case 4}: $p$ accesses a vertex gadget $J^\prime$ which interfaces
      $J$ at the red edge set.  Each time that $p_J$ or $p_{J^\prime}$
      traverses a grue edge before the last vertex in $p$ that belongs to both
      $J$ and $J^\prime$, another edge of the same color must be accessed, by
      hypercube property (3).  Since there must be an even number of edges of
      any color, this means that unrepeated edges in the consistency gadget
      must be repeated by similar edges in the portion after this last
      crossover.  Also, since all but at most one shade of red has to be
      traversed twice within the non-consistency-gadget portion, all the red
      colors together bring the path to within one edge of the border of a
      vertex gadget.  But in that case there are at least two grue edges in the
      consistency gadget portion, again pushing the path beyond the radius of
      the vertex gadget.

      Thus there is no way for such a $p$ to exist.
    \end{proof}

    This means that the resulting graph may be $4k$-nonrepetitively
    $6k$-edge-colored iff the original graph can be 3-edge-colored.  Note that
    there is nothing special about the original graph being of degree 3.  If
    the original graph is regular of degree $d$, then we can construct a
    similar new graph with a vertex gadget being a slice of a $2dk$-hypercube,
    a consistency gadget being a modified $2(d-1)k$-hypercube, and so on.
  \end{proof}

  Taking into account the vertex and consistency gadgets, for an original graph
  with $a$ edges, the resulting graph of this reduction has $\Theta(
  \frac{2a}{d} \cdot 2^{2dk}+a \cdot 2^{2(d-1)k})=\Theta(\frac{a}{d}2^{2dk})$
  edges.  This means that if $k$ depends logarithmically on $a$, then the
  reduction still runs in polynomial time with respect to $a$.  This raises
  questions about other versions of the problem where $k$ is not constant, but
  a sub-linear function of $\lvert G \rvert$.

  \section{THUE NUMBER}
  We now proceed to the main result of the paper, showing that \textsc{Thue
    number} is $\mathbf{\Sigma_2^p}$-complete.  We use a reduction based on
  that of Theorem 3, taking a graph very similar to that produced in Theorem 3
  and replace the edges with various gadgets depending on their function in
  that reduction.  We also add constraint gadgets in order to force the edge
  gadgets that would have been colored the same in Theorem 3 to still have
  similar patterns of colors.  After we show that a nonrepetitive coloring of
  this graph must satisfy a large number of constraints stemming from this
  local structure, we can use the global structure from Theorem 3 to show that
  an actual nonrepetitive coloring exists if and only if the ancestral instance
  of co-$\forall\exists$3SAT is positive.

  While this is the idea of the proof, in reality, the local structure is not
  cleanly separated from the global structure: the format of the latter helps
  constrain the former.

  \begin{thm}
    \textsc{Thue number} is $\Sigma_2^p$-complete.
  \end{thm}
  \begin{proof}
    Given an instance $\forall x \exists y\,f(x,y)$ of co-$\forall\exists$3SAT,
    we will refer to the instance $(G,\mathcal{S})$ of \textsc{restricted Thue
      number} constructed in Theorem 3.  Let $c=\lvert\bigcup\mathcal{S}\rvert$
    be the number of colors coloring $G$, $u=\lvert y \rvert$, and let $\ell$
    be the smallest integer such that $2^{\ell} \geq c+u+1$ and $m=4\ell+3$.

    Let $w_i^0,w_i^1$ be the two colors corresponding to the universally
    quantified variable $x_i$.  Then we modify $G$ by introducing new colors
    $w_i^I$ and $w_i^F$ and replacing each edge colored by $A \subseteq
    \{w_i^0,w_i^1\}$ by three edges in series colored $\{w_i^I\}A\{w_i^F\}$.
    Clearly, this graph $(G',\mathcal{S}') \in$ \textsc{restricted Thue number}
    iff $(G,\mathcal{S})$ is.  We construct from this graph an instance $(H,
    2^m+6)$ of \textsc{Thue number}.

    Let $S^*$ be the set of colors that are not $w_i^j$ for $j \in 2$.  For
    every edge $e$ of $G'$ whose color is from $S^*$, we add to $H$ a
    7-hypercube $E^e$ with two opposing vertices ($v_0$ for the one closer to
    the tip of the snout along a potentially square path and $v_1$ for the
    other one) acting as the vertices of the edge gadget.  For each color $s
    \in S^*$ and for each pair $X_i=\{x_i^0,x_i^1\}$, we also add a consistency
    gadget $Q^s$ (respectively $Q_X^i$).  This consists of a $2^m$-clique with
    every vertex identified either with the central vertex of the first two
    layers of a 7-hypercube, called the \emph{plume} emanating from this
    vertex, or with a vertex of an edge gadget.  For each edge $e$ whose color
    is $s$, two vertices of $E^e$ are identified with vertices of $Q^s$: $u_0$
    for one that is distance 3 from $v_0$ and $u_1$ for one that is distance 4
    from $v_0$.
    \begin{figure}
      \begin{center}
	\includegraphics{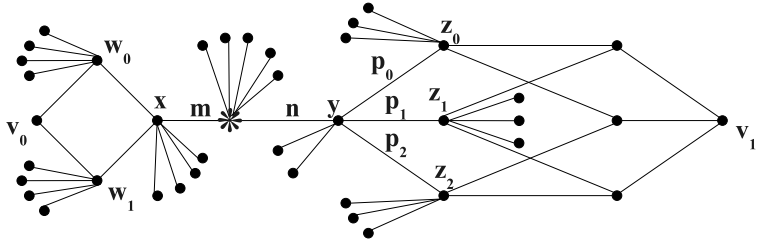}
      \end{center}
      \caption{A $C$-gadget.  It may be thought of as a sequence of
	strung-together hypercubes of dimensions 2, 1, 1, and 3, together with
	plumes that inflate the degree of some vertices.}
    \end{figure}
    \begin{figure}
      \begin{center}
	\includegraphics{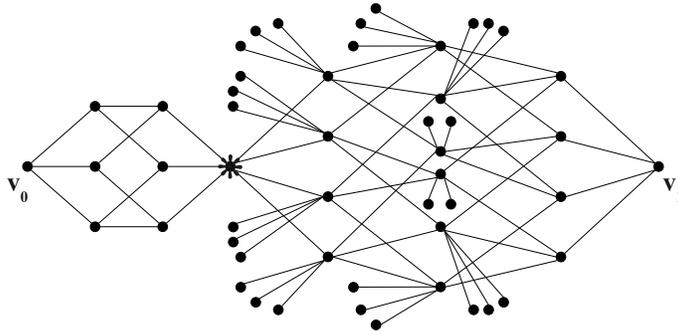}
      \end{center}
      \caption{A $C$-gadget.  It may be thought of as two hypercubes of
	dimensions 3 and 4 together with plumes that inflate the degree of some
	vertices.}
    \end{figure}

    For the edges colored by subsets of $\{x_i^0,x_i^1\}$, the choice edge in
    the snout becomes the gadget depicted in Fig.~4, which we call a
    \emph{$C$-gadget} or $C^i$.  The negative edge in the variable gadget
    becomes the gadget depicted in Fig.~5, which we call an \emph{$N$-gadget}
    or $N^i$.  In both cases, the asterisked vertex $u$ is identified with a
    vertex of $Q^{X_i}$ and $v_0$ and $v_1$ are connected to other edge
    gadgets.

    Finally, the other edge becomes a \emph{$P$-gadget} $P^i$.  This is a
    7-hypercube attached to the consistency gadget in the same way as $E^e$,
    but is modified in several ways.  In order to explain this, let us first
    label the dimensions of the hypercube $a,b,c,d,ab,ac,ad$.  (This labeling
    will be explained later.)  Let $a,b,ab$ span the distance from $u_0$ to
    $v_0$, and, correspondingly, let $c,d,ac,ad$ span the distance from $u_0$ 
    to $v_1$.  Let us eliminate the edges $c,d,ac,ad$ adjacent to $v_1$.  We
    then nearly saturate the vertices that are separated from $u_0$ by the
    paths $(a,c,ac)$ and $(a,d,ad)$ (call them $u_{-1}$ and $u_{-2}$
    respectively) with $2^m-2$ additional edges, but then remove the original
    edges corresponding to $b$ from both and $c$ and $d$ from $u_{-1}$ and
    $u_{-2}$ respectively.

    Now that we have constructed a graph, we will attempt to color it
    nonrepetitively with $k$ colors, and show that it has a nonrepetitive
    coloring iff the instance of $\forall\exists$3SAT is negative.  Let us call
    the set of colors $\mathcal{K}$ and call the color of an edge $e$ under a
    hypothetical coloring $C(e)$.  Conversely, we denote by $[c]$ an edge of
    color $c$ if its more precise identity is irrelevant.

    \begin{claim2}
      Let $e_1,\ldots,e_n$ be the set of edges of color $s \in S^*$.  Then
      in a nonrepetitive coloring, $E^{e_i}$ must all be colored with the same
      7 colors.  Similarly, for every $i$, $C^i$, $N^i$, and the hypercube
      portion of $P^i$ must all be colored with the same 7 colors.
    \end{claim2}
    \begin{proof}
      Suppose we have a nonrepetitive coloring $C(e)$.

      Inside each consistency gadget $Q$, each vertex of the $2^m$-clique is
      saturated.  Let $u$, $v$, and $w$ be such vertices, and suppose that
      $C(uv)$ is the same as the color of the edge of a plume adjacent to $w$.
      Then there is an edge of color $C(uw)$ adjacent to $v$.  Together with
      $uw$, these four edges then form a square path.  To avoid this, the set
      $Z(Q)$ of 7 colors that colors the first layer of each plume must be
      disjoint from the set $G^*(Q)$ of $2^m-1$ colors that colors the clique.
      By a similar argument, the second layer of every plume must also have a
      color in $Z(Q)$.  According to \cite{alon}, the coloring can be
      constructed in such a way that $G^*(Q)$ consists of the nonzero elements
      of the group $Z_2^m$.  (From now on we will refer to just $Z$ and $G^*$
      when $Q$ is clear from the context.)

      By the same token, the edges in the edge gadgets of distance at most two
      from $u_i$ must be colored with colors from $Z$.  Now suppose we have an
      edge $e$ that has distance 3 from $u_i$.  If $C(e) \in G^*$, then the
      path consisting of $e$ and the two edges that connect it to $u_i$ can be
      duplicated inside the consistency gadget, so $C(e) \in Z$.

      The maximum distance of an edge in any edge gadget from the closer $u_i$
      is 4.  Let us first consider a gadget $E^e$, and let $Z=\{c_1,\ldots,
      c_7\}$.  We know that edges in the first three layers from $u_0$ have
      colors in $Z$, so they must be arranged according to the hypercube lemma.
      Let $w_{ijk}$ be separated from $u_0$ by edges $[c_i][c_j][c_k]$.  Let
      $e_4$ be the edge incident to $w_{123}$ that would be colored $c_4$ if
      $E^e$ were nonrepetitively colored by $Z$, and let $e_3$ be the edge
      incident to $w_{124}$ that would be colored $c_3$.  These edges are
      adjacent.

      Now we have four cases, since both $C(e_3)$ and $C(e_4)$ can be either in
      $Z$ or in $G^*$.  Suppose that $C(e_3), C(e_4) \in Z$.  Then if it is not
      true that $C(e_3)=c,C(e_4)=d$, and we don't have a trivial repetition,
      then we have a square path of the form $[c_1][c_2]e_4e_3[c_1][c_2]
      [C(e_4)][C(e_3)]$.  So suppose that $C(e_4) \in G^*$.  But we have a
      square path $[C(e_3)][C(e_4)][c_2][c_4][c_1]e_3e_4[c_2][c_4][c_1]$ that
      starts with two edges of the consistency gadget.  Finally, if $C(e_4) \in
      Z$ but $C(e_3) \in G^*$, we have the analogous path $[C(e_4)][C(e_3)]
      [c_2][c_3][c_1]e_4e_3[c_2][c_3][c_1]$.

      Hence in a nonrepetitive coloring the fourth layer is also colored by $Z$
      according to the hypercube lemma.  Since we can prove the same thing
      starting from $u_1$, this means the colorings from the two $u_i$s match,
      and so the entire hypercube is colored in this way.

      This is also valid for the $P$-gadget.  Let us start from $u_1$.  None of
      the removed edges has distance less than 4 from this vertex, and six have
      distance 4.  Since none of them are joined at a vertex whose distance
      from $u_1$ is 4, we can apply the previous argument to the ones that are
      not removed.  This leaves only the edges whose hypercube distance from
      $u_0$ is at most 3; but since no removed edges have distance less than 3
      from $u_0$, this is the same as their actual distance, so we are done.

      If the color of the edge in $Q^X$ between $u_0$ and $u_1$ is represented
      in the plume of $u_{-j}$, there is a square path of length 8 going
      through $u_{-j}$, $u_0$, and $u_1$.  Hence this must be the one color in
      $G^*$ that is missing from the plume.

      In the other two gadgets, all the edges except the clump closest to $v_1$
      have distance at most 3 from the vertex $u$ attaching them to $Q^X$, and
      thus are known to have colors from $Z$.

      Now, in $P^X$ we have paths $\tau_0=[a][b][ab],\tau_1=[c][d][ac][ad],
      \tau_2=[b][c][ab][ac]$, $\tau_3=[b][d][ab][ad]$ each between $u_{\pm j}$
      and $v_0$, and paths $\sigma_0=[a][b][ab],\sigma_1=[c][ac][a],\sigma_2=
      [d][ad][a]$ each between $u_{\pm j}$ and $v_1$.  Now, vertex $v_0$ is
      also vertex $v_1$ of an $E$-gadget $E^{e'}$, and the color of $e'$ in $G$
      is $x_i^i$.  Some path $[k_0][k_1][k_2]$ leads us to the vertex $u_1$ of
      this gadget.  Now, this means that we have a path $[k_2]\tau_i[k_0][k_1]
      [k_2][C(\tau_i)][k_0][k_1]$ through the two gadgets and $Q^{x_i^i}$,
      which a nonrepetitive coloring would require to have a loop.  Hence in
      the group $G=(G^*(Q^{x_i^i}) \cup \{0\},+) \cong Z_2^m$ we must have $a+
      b+ab=0$.  Furthermore, since a possible $T_i$ is $(d,b,ab,ad)$, having a
      loop of size 3 means that $b+ab+d=0$ or $b+ab+ad=0$, which contradicts
      the previous statement.  Hence we have $b+d+ab+ad=0$, and by the same
      argument $b+c+ab+ac=0$.  By addition we then have the statements $c+d+ac+
      ad=0$, $a+c+ac=0$, and $a+d+ad=0$.  It is easy to see that we can create
      no more such identities using colors from $Z(Q^X).$  By a similar
      argument, each $\sigma_i$ must form a loop in $Q^{x_i^f}$, which gives us
      the same six identities.

      By construction, the other two gadgets attached to $Q^X$ are also
      preceded by $x_i^i$ edge gadgets, so each path from $v_0$ to $u$ in these
      gadgets must have one of the three sets of colors $\{a,b,ab\}$, $\{a,c,
      ac\}$, $\{a,d,ad\}$.

      Now, in $N^X$, the edges whose distance from $u$ is at most 3 are known
      to be in $Z(Q^X)$.  Furthermore, the two paths $[ac][c][a]$ and $[ad][d]
      [a]$, if they were present starting from $u$, would create a square path
      where an edge adjacent to some $u_{-j}$ of $P^X$ repeats the edge between
      $u_0$ and $u$.  Now, if $A$ is the set of colors coloring the cube in
      $N^X$, it must be that $\sum A=0$ in $G(Q^{x_i^i})$.  Thus it must be
      that $A=\{a,b,ab\}$ and that the vertices of the 4-hypercube that are
      adjacent to $u$ and those that are distance 2 away are effectively
      saturated.  Hence by the saturation lemma all of the 4-hypercube, except
      maybe the edges adjacent to $v_1$, is colored from $\{c,d,ac,ad\}$.  But
      to prevent a square path through $Q^{x_i^f}$, these last edges must also
      be colored from that set.

      We apply a slightly different argument to $C^X$.  To prevent a repetition
      going through $u_{-j}$ and $u_1$ of $P^X$, either the set $A_0$ of colors
      of the paths from $v_0$ to $u$ must be $\{a,ab,b\}$, or the edge labeled
      $n$ must be colored $b$.  Furthermore, if $A_0=\{a,ab,b\}$, to prevent a
      repetition going through $u_{-j}$ and $u_0$, the edge labeled $n$ must be
      colored $c$ or $d$, and otherwise the edge labeled $m$ must be colored
      $c$ or $d$.  The plumes at $w_i$ and $y$ clearly cannot contain an edge
      of the same color as $m$, and the plumes at $z_i$ cannot contain an edge
      of the same color as $n$.  Thus these vertices are effectively saturated,
      and we can apply the saturation lemma to the faces of the cube closer to
      $u$.  Furthermore, if some $p_j$ is the same color as an edge in the
      square, then we have a square path of length 6 through the corresponding
      $w_i$ and $z_k$.  Thus except for the edges adjacent to $v_1$, we know
      that the paths from $u$ to $v_1$ must be colored from the set $A_1=Z(Q^X)
      \setminus A_0$.  Now, by the argument above, $G(Q^{x_i^i})$ must have
      $\sum A_0=0$, and thus so does $G(Q^{x_i^f})$.  So $\sum A_1=0$ in $G(
      Q^{x_i^f})$.  Then to prevent a square path through $Q^{x_i^f}$, the last
      edges of paths from $u$ to $v_1$ must also be colored from $A_1$.
    \end{proof}
    \begin{claim2}
      Given a nonrepetitive coloring of $H$, for every existentially quantified
      variable $x$, there is a path from $v_0$ to $v_1$ of $C^X$ that has the
      same pattern of colors as a path from $v_0$ to $v_1$ of either $N^X$ or
      $P^X$, but not both.
    \end{claim2}
    \begin{proof}
      Let us assume the coloring of $P^X$ that we have implied: we have already
      shown that this is unique up to permutation.  We have also already shown
      that given this coloring, $N^X$ must have colors $a,b,ab$ in the cube and
      $c,d,ac,ad$ in the hypercube, and that $C^X$ must have one of the
      zero-sum sets of 3 colors on the shorter side and one of the zero-sum
      sets of 4 colors on the longer side.

      If the shorter set is $\{a,b,ab\}$ then clearly there is a path through
      the gadget that is the same as a path through $N^X$.  On the other hand,
      there is not a path that is the same as a path through $P^X$, since no
      path through $P^X$ ends in $c,d,ac$, or $ad$.

      Otherwise the gadget contains either the path $(a,ac,c,b,d,ad,ab)$ or the
      path $(a,ad,d,b,c,ac,ab)$.  The removed edges in the $P$-gadget are all
      in either fourth or last position.  Now, the edge $ab$ adjacent to $v_1$
      is present.  Furthermore, the removed $b$-edges are preceded by the sets
      $\{ab,c,ac\}$ and $\{ab,d,ad\}$, so they are not the $b$-edges in these
      paths.  Hence, these paths are present in $P^X$.  Since the first three
      colors cannot be a permutation of $\{a,b,ab\}$, there is no path that is
      the same as a path through $N^X$.
    \end{proof}
    \begin{claim2}
      If $\forall x \exists y\,f(x,y)$, then every coloring of $H$ has a square
      path.
    \end{claim2}
    \begin{proof}
      Let us assume that we have a nonrepetitive coloring.  It must then agree
      with claims 1 and 2.  From claim 1, the path $c_1 \ldots c_7$ from $v_0$
      to $v_1$ exists for all the $E$-gadgets.  On the other hand, claim 2
      tells us that a nonrepetitive coloring of a $C$-gadget is effectively a
      choice between the corresponding $N$-gadget and $P$-gadget.  So by
      choosing a $y$ that satisfies $f$ for the choice of $x$ provided by the
      coloring, we can create a long square path analogous to the one for $G$.
    \end{proof}
    Now, assuming that $\exists x\forall y\,\not f(x,y)$, we construct a
    coloring we claim to be nonrepetitive.  Let $x_1,\ldots,x_p$ be an
    assignment to the universally quantified variables such that $\exists y\,f(
    x,y)$ is not satisfiable.  We set aside 7 colors $a_i,b_i,c_i,d_i,ab_i,
    ac_i,ad_i$ for each consistency gadget $Q^{w_i}$, all of them distinct.
    We color the $E$-gadgets so that for one, $\{a_i,b_i,ab_i\}$ spans the
    distance from $u_0$ to $v_0$ and $\{c_i,d_i,ac_i,ad_i\}$ spans the distance
    from $u_0$ to $v_1$; for another, $\{a_i,c_i,ac_i\}$ and $\{b_i,d_i,ab_i,
    ad_i\}$; and for the third, $\{a_i,d_i,ad_i\}$ and $\{b_i,c_i,ab_i,ac_i\}$.
    If $G'$ has a fourth edge of that color, then we make sure that the
    coloring of the gadget corresponding to the edge in the snout is different
    from all the others, but color the fourth edge gadget in one of these three
    ways.

    For a consistency gadget that controls a universally quantified variable
    $x_j$, we color $N^{x_j}$ and $P^{x_j}$ in the way already described.  We
    then color $C^{x_j}$ so that it contains the path $(a_i,b_i,ab_i,c_i,d_i,
    ac_i,ad_i)$ if $x_j=1$ and the path $(a_i,ac_i,c_i,b_i,d_i,ad_i,ab_i)$ if
    $x_j=0$.

    Let us set aside a set of 7 colors $\beta_1,\ldots,\beta_7$ distinct from
    the ones enumerated so far.

    Now we need to color the consistency gadgets, which we do using the group
    method in \cite{alon}.  First, we assign the elements of the subgroup of
    $Z_2^m$ generated by the first three generators to the set containing each
    edge gadget's $u_0$ and $u_1$, whose cardinality clearly is at most 8.
    Furthermore, we set the numbers of the vertices so that the sum of the
    colors of $u_0$ and $u_1$ within an edge gadget does not depend
    on the edge gadget.  Thus the colors of the edges between these vertices
    are nonzero members of this subgroup, and the edge between $u_0$ and $u_1$
    is the same color for each edge gadget.  We assign to these edges the
    colors $\beta_1,\ldots,\beta_7$, with $\beta_1$ being the color between
    $u_0$ and $u_1$ for every $E$- and $P$-gadget.

    For a consistency gadget corresponding to a universally quantified variable
    $x_i$, we label the vertex $u_0$ of $P^{x_i}$ with 0 and the other three
    $u$-vertices with a generator of the subgroup.  This ensures that all the
    edges that connect two $u$-vertices are colored differently.

    Now, using the rest of the generators, we can find $2^\ell-1 \geq c+2u$
    subgroups $G^g$ of $Z_2^m$ isomorphic to $Z_2^4$, one for each nonzero
    element $g \in Z_2^{m/4}$, by taking the generators $g^0,g^1,g^2,g^3$ to
    have $i$th digit $g^j_i=g_{(j-i-3)/4}$ if $4 \mid j-i-3$ and 0 otherwise.

    We next arbitrarily assign group elements to the rest of the vertices, and
    hence edges.  However, we assign colors to the group elements in such a way
    that for each set $Z(Q^i)$ we take a distinct subgroup of the kind
    described and assign $a_i,b_i,c_i,d_i$ to the generators and $ab_i,ac_i,
    ad_i$ to the elements $a_i+b_i,a_i+c_i$, and $a_i+d_i$, respectively.

    Since the first three digits are zero for all elements of these sets, we
    have guaranteed that we can never get $\beta_i$ by adding together colors
    of edge gadgets.

    \begin{claim2}
      This coloring is nonrepetitive when restricted to a consistency gadget
      $Q$ together with its edge gadgets.
    \end{claim2}
    \begin{proof}
      We have already shown that the coloring of each gadget is nonrepetitive.

      A path within $Q$ whose colors are contained in $Z(Q)$ cannot reach a
      vertex at a distance of more than 2 from the interface with colors in
      $G^*$.  On the other hand, by construction, no colors in $G^*(Q)$ are
      present in the edge gadgets within distance 2 of $u_i$.  Thus a portion
      of a square path within an edge gadget must be repeated at least
      partially in another edge gadget.

      For an $E$-gadget, this means that one such portion must not be at the
      end of a path.  Thus, since the path is open, it goes through the edge 
      gadget from $u_0$ to $u_1$, and contains an odd number of edges of each
      color in $Z(Q)$.  If some such portion is repeated by sections at both
      the beginning and end of the path, then one of these sections must be in
      an $E$-gadget, and therefore we can splice the other section onto that
      one to create another open square path.  Thus we can build a shorter
      square open path by replacing each such portion with the edge between
      $u_0$ and $u_1$, which we know has the same color $\beta_1$ for each
      gadget.  But this substitution produces a square path entirely inside
      $Q$, which must be a loop.

      Assume now that $Q$ corresponds to a universally quantified variable.  We
      have already constructed our coloring so that a path starting from one of
      the edges adjacent to $u_{-j}$ and continuing to either $C$ or $N$ cannot
      repeat.  Now, we did not color an edge adjacent to $u_{-j}$ by $\beta_1$,
      so a path from $u_0$ to $u_1$ that is repeated by edges within $P$ must
      be at least two edges long.  But a path from $u_{-j}$ in colors from
      $G^*(Q)$ may only be one edge long.  Now, the set of colors separating
      $u_{-j}$ and $u_i$ is different from the set of colors separating
      $u_{-3+j}$ and $u_{1-i}$.  Therefore, we cannot find a square path that
      traverses both.  So any square path must contain at least two sections
      within $Q$.  Such sections, if they are each between two $u$-vertices,
      cannot repeat each other since by our construction the group element
      leading between each pair of $u$-vertices is different.  But such a path
      also cannot use the plumes because any vertex in $N$ and $C$ has distance
      at most 4 from $u$, so this together with the plumes would not be enough
      to repeat a path through $P$ containing 7 distinct colors.
    \end{proof}
    \begin{claim2}
      In a square path, colors in edge gadgets cannot be duplicated by the same
      colors in consistency gadgets, except perhaps in one plume.
    \end{claim2}
    \begin{proof}
      Suppose we have a square path that serves as a counterexample to this,
      with some sequence of colors occurring in an edge gadget $E_0$ repeated
      by a sequence in a consistency gadget $Q_1$.  We give the name $p$ to
      this path and $p_\Gamma$ to the portion of it that goes through a gadget
      $\Gamma$.  We orient $p$ so that $p_{E_0}$ comes before $p_{Q_1}$; we
      therefore refer to the first and second halves of $p$, and to ends and
      beginnings of sections.

      We start with \emph{case 1}: $E_0$ does not belong to $Q_1$ and $p_{E_0}$
      is not at the beginning of $p$.  Then the ends of $p_{E_0}$ are vertices
      from the set $\{u_{\pm j}\} \cup \{v_i\}$.  By construction, the colors
      of any such path form a loop in $Q_1$, so the path cannot be open.

      Thus if $E_0$ does not belong to $Q_1$, then $p_{E_0}$ is the beginning
      of $p$.  Then the edges following $p_{E_0}$ must come either from a
      consistency gadget $Q_0$ (\emph{case 2}) or another edge gadget $E^*$
      (\emph{case 3}.)

      \emph{Case 2}: In order to link the two half-paths, the first half must
      eventually leave $Q_0$.  If it leaves early, so that an adjacent edge
      gadget is repeated in $Q_1$, then we are back in case 1.  If it leaves
      late, so that it is repeated by an edge gadget adjacent to $Q_0$, then we
      are also back in case 1: if the path ended in that edge gadget, then
      there would be nothing to repeat the edge gadgets between $Q_0$ and
      $Q_1$.  Finally, if the vertex $u_i$ at which the path leaves $Q_0$ is
      repeated by a vertex $u_j$ at which the path leaves $Q_1$, then since we
      are in different consistency gadgets, the edges we traverse have
      different sets $Z$ of colors, a contradiction.

      \emph{Case 3}: If $E^*$ does not belong to $Q_1$, then it must be
      duplicated inside $Q_0$, since the other option is an edge gadget
      belonging to $Q_0$, which does not match colors.  Therefore, substituting
      $E^*$ for $E_0$, we get case 1 and a contradiction.

      So $E^*$ belongs to $Q_1$.  This means that the end of $p_{E_0}$ is
      $v_i$.  Furthermore, since $p_{E^*}$ must have odd numbers of at least
      three colors, and since $p_{E_0}$ is the beginning of $p$ and therefore
      all of $p_{E^*}$ is in the first half, $p_{E^*}$ cannot be repeated in a
      plume of $Q_1$.  Thus $Q_1$ must end with $u_j$ of an edge gadget $E_1$
      which repeats $v_i$.  If $p_{E^*}$ ends with $v_{1-i}$, then $p_{E_1}$
      ends with $u_{1-j}$ and we are again in case 1.  On the other hand, if
      $p_{E^*}$ ends at $u_k$, then $P_{E_1}$ ends with $v_\ell$, and so either
      we are back in case 1 or the path ends after $P_{E_1}$.  But that would
      mean that the colors of $p_{E_0}$ and $p_{E_1}$ would add up in $G(Q_1)$
      to $\beta_i$ for some $i$, which contradicts the construction.

      The remaining possibility, \emph{case 4}, is that $E_0$ belongs to $Q_1$.
      Then for colors to match, the part of $p_{Q_1}$ that repeats $p_{E_0}$
      must be contained in a plume of $Q_1$ and therefore be the end of the
      path.
    \end{proof}
    \begin{claim2}
      No square path in this coloring goes through a consistency gadget.
    \end{claim2}
    \begin{proof}
      By claim 5, colors in edge gadgets must be duplicated by colors in edge
      gadgets.  Now, suppose we have a square path $p$ that goes through a
      consistency gadget $Q$.  We use the notation of claim 5 to denote
      portions of this path.  On both sides of $p_Q$, we have paths through
      edge gadgets $E_1$ and $E_2$.  We now enumerate cases once again.

      \emph{Case 1} occurs when $p_{E_1}$ is a path between $u_0$ and $u_1$.
      Then there must be a subpath repeating it, which, by claim 5, then also
      goes from $u_i$ to $u_{1-i}$ in an edge gadget $E_1^*$ of $Q$.  We can
      create another square path by replacing $p_{E_1}$ and $p_{E_1^*}$ by
      single edges within $Q$, reducing the problem to any of the other cases.

      So we can assume that both $p_{E_1}$ and $p_{E_2}$ go to a $v_i$.  Then
      if we orient the path, one subpath goes from an edge gadget to $Q$ and
      the other goes from $Q$ to an edge gadget.  Therefore, by claim 5, one
      cannot be the repetition of the other.  Thus $p$ must traverse $Q$ twice
      and associated edge gadgets four times, at least three of these times
      passing between a $u_i$ and a $v_i$.  (The fourth time it does not have
      to do this since a traversal may be split between the beginning and the
      end of $p$.)  Furthermore, the edge subpaths must be in two pairs, each
      of which has one set of colors, and at least one of the pairs must have
      both its subpaths continue in other edge gadgets.

      Suppose now that $Q$ is not associated with a universally quantified
      variable.  \emph{Case 2} occurs when the two subpaths are in the same
      $E$-gadget.  Since no edge in $G'$ is adjacent to edges of one color
      through both incident vertices, the colors of the neighboring edge
      gadgets must be different and the path cannot be square.

      In \emph{case 3}, the two subpaths are in different $E$-gadgets.  In
      order for them to be colored the same, we must assume that $Q$ controls
      four $E$-gadgets and neither of the $E$-gadgets traversed by the subpaths
      corresponds to an edge in the snout of $G'$.  But the edges not in the
      snout all border different colors, so the path once again cannot be
      square.

      \emph{Case 4}: We are left with the possibility of $Q^{x_j}$ being
      associated with a universally quantified variable $x_j$.  We colored
      $Q^{x_j}$ in such a way, however, that the sets of paths within it
      between every pair of $u$-vertices are disjoint.  Therefore, if we orient
      $p$ in a certain way, the second $p_{Q^{x_j}}$ must be at the end of the
      path.  By the same token, $p$ must begin with a section in one of the
      edge gadgets belonging to $Q^{x_j}$.

      All the $v_0$s of the edge gadgets of $Q^{x_j}$ are incident to
      $E$-gadgets of colors $Z(Q^{x_j^i})$ and all the $v_1$s are incident to
      $E$-gadgets of colors $Z(Q^{x_j^f})$.  Hence paths to $v_1$ can be
      repeated by paths to $v_1$ and paths to $v_0$ can be repeated by paths to
      $v_0$.  By construction, the possible repetitions of paths between $u_j$
      and $v_i$ are then
      \begin{enumerate}
      \item if $x_j=0$, a path between $v_0$ and $u$ in $N^{x_j}$ with a path
	between $v_0$ and $u_0$ in $P^{x_j}$
      \item if $x_j=1$:
	\begin{enumerate}
	  \item A path between $v_0$ and $u$ in $C^{x_j}$ with a path between
	    $v_0$ and $u$ in $N^{x_j}$ and a path between $v_0$ and $u_0$ in
	    $P^{x_j}$
	  \item A path between $v_1$ and $u$ in $C^{x_j}$ with a path between
	    $v_1$ and $u$ in $N^{x_j}$.
	\end{enumerate}
      \end{enumerate}
      In each case, therefore, there can be only one such repetition in an open
      path.  Since gadgets belonging to $Q^{x_j}$ occupy both the beginning and
      the end of $p$, there may be no traversal of consistency gadgets other
      than $Q^{x_j}$ within $p$.

      A path through $N^{x_j}$ cannot repeat a path through $P^{x_j}$ since
      these paths would have to go through edge gadgets corresponding to
      different direction-determining colors of $G'$.  Also, a path going
      through $v_0$ of $C^{x_j}$ is going towards the tip of the snout, and
      therefore cannot meet any more gadgets belonging to $x_j$.  This gives us
      a contradiction if $v_0$ is part of one of the square sections.  Thus we
      are left with possibility 2(b).

      This means that there must be a subpath consisting only of edge gadgets
      from $v_1$ of $C^{x_j}$ to some other $v_i$ of an edge gadget belonging
      to $Q^{x_j}$.  If the path is from $C^{x_j}$, then it must traverse an
      edge gadget belonging to $Q^{a_i^1}$, which must be repeated by an edge
      gadget at the tip of the snout, which $p$ cannot reach.  
    \end{proof}
    This restricts any potential square paths to going through a chain of edge
    gadgets.  But this means that a square path maps to a square path in $G'$,
    which does not exist if $\neg \forall x\exists y\,f(x,y)$.  Hence this is a
    reduction.

    It is easy to see that the reduction runs in polynomial time, specifically
    $O(n^8)$.
  \end{proof}

  \section{Acknowledgements}
  We would like to thank Chris Umans for his mentoring and helpful suggestions
  at every stage of the research; and Grigori Mints and Yuri Manin for
  reviewing the proofs.

\end{document}